\documentclass[10pt,letterpaper,twocolumn,english,conference]{IEEEtran}
\usepackage[T1]{fontenc}
\usepackage[latin9]{inputenc}
\usepackage{amsmath}
\usepackage{graphicx}
\usepackage{amssymb}

\makeatletter
\newtheorem{thm}{Theorem}

\usepackage{psfrag}

\DeclareMathOperator{\argmin}{argmin}
\DeclareMathOperator{\rank}{rank}
\DeclareMathOperator{\rankdef}{rankdef}

\makeatother

\usepackage{babel}

\begin{document}

\title{Multishot Codes for Network Coding\\
using Rank-Metric Codes}

\author{\IEEEauthorblockN{Roberto W. Nóbrega and Bartolomeu F. Uchôa-Filho} \IEEEauthorblockA{Communications
Research Group\\
Department of Electrical Engineering\\
Federal University of Santa Catarina\\
Florianópolis, SC, 88040-900, Brazil\\
\{rwnobrega, uchoa\}@eel.ufsc.br}}
\maketitle
\begin{abstract}
The multiplicative-additive finite-field matrix channel arises as an adequate model
for linear network coding systems when links are subject to errors and erasures,
and both the network topology and the network code are unknown. In a previous work
we proposed a general construction of multishot codes for this channel based on the
multilevel coding theory. Herein we apply this construction to the rank-metric space,
obtaining multishot rank-metric codes which, by lifting, can be converted to codes
for the aforementioned channel. We also adapt well-known encoding and decoding algorithms
to the considered situation.
\end{abstract}

\section{Introduction\label{sec:introduction}}

Noncoherent linear network coding with unreliable links in a multicast scenario has
been given a good deal of attention since the seminal work of Koetter and Kschischang
\cite{koetter-kschischang}. The problem is suitably modeled by the \emph{multiplicative-additive
finite-field matrix channel} given by\begin{equation}
\mathbf{Y}=\mathbf{A}\mathbf{X}+\mathbf{Z},\label{eq:matrix-channel}\end{equation}
where $\mathbf{X}$ is the $N_{\mathrm{i}}$-by-$T$ channel input matrix, $\mathbf{Y}$
is the $N_{\mathrm{o}}$-by-$T$ channel output matrix, $\mathbf{A}$ is the $N_{\mathrm{o}}$-by-$N_{\mathrm{i}}$
multiplicative transfer matrix, and $\mathbf{Z}$ is the $N_{\mathrm{o}}$-by-$T$
additive error matrix. All matrices are over some finite field~$\mathbb{F}_{q}$.

In the context of network coding, $N_{\mathrm{i}}$ stands for the number of packets
transmitted by the source node while $N_{\mathrm{o}}$ stands for the number of packets
received by a given sink node at each time slot, a packet consisting of $T$~symbols
from $\mathbb{F}_{q}$; the matrices $\mathbf{X}$ and $\mathbf{Y}$ are then formed
by juxtaposing the $N_{\mathrm{i}}$ transmitted and the $N_{\mathrm{o}}$ received
packets, respectively, seen as row vectors. Matrix $\mathbf{A}$ represents the network
transfer matrix, which depends on the network topology and network code used (both
assumed unknown in a noncoherent scenario), and matrix $\mathbf{Z}$ is related to
link errors%
\footnote{If $N_{\mathrm{e}}$ error packets are injected into the network, and if we dispose
these packets to form an $N_{\mathrm{e}}$-by-$T$ matrix $\mathbf{Z}'$, then we
have that $\mathbf{Z}=\mathbf{D}\mathbf{Z}'$, for some $N_{\mathrm{o}}$-by-$N_{\mathrm{e}}$
transfer matrix $\mathbf{D}$. This decomposition, although meaningful in a general
context, is unimportant here.%
}. For simplicity, we set $N_{\mathrm{i}}=N_{\mathrm{o}}=N$ throughout this work.

To achieve a reliable communication over this channel, \emph{matrix codes} are employed.
So far, attention has mostly been given to \emph{one-shot} matrix codes, that is,
codes that use the matrix channel~\eqref{eq:matrix-channel} only once. In this
case, a code is simply a non-empty subset of matrices. Among the existing constructions
for one-shot codes, we highlight that of Silva \textsl{et al.}~\cite{silva-rank-metric},
in which matrix codes are obtained from \emph{rank-metric codes}. Rank-metric codes,
in turn, were already studied before (\textsl{e.g.}, by Gabidulin~\cite{gabidulin}
and Roth~\cite{roth-crisscross}) in distinct contexts.

In contrast, we herein consider blocks of $n$ consecutive uses of the matrix channel~\eqref{eq:matrix-channel}.
Under this framework, a code is now a non-empty subset of $n$-tuples of matrices;
we call it an \emph{$n$-shot }(or \emph{multishot})\emph{ }matrix code. Multishot
codes for network coding were studied in~\cite{multishot}, where a general construction
based on the well-known \emph{block coded modulation} and \emph{multilevel code construction}
of Imai and Hirakawa~\cite{imai-hirakawa} has been proposed. It has been shown
in~\cite{multishot} that multishot codes can correct more errors than one-shot
codes, which motivates this work.

In the present paper, we combine the rank-metric approach of Silva \textsl{et al.}
with the multilevel code construction of Imai and Hirakawa to obtain multishot codes
for network coding. The problem can also be interpreted within the theory of \emph{generalized
concatenated codes} of Blokh and Zyablov~\cite{generalized-concatenated}, and by
doing so it is possible to adapt well-known encoding and decoding procedures to the
case at hand (\textsl{cf.}~\cite[Chapter~15]{lin-costello}).

We begin in Section~\ref{sec:background} by providing background on concepts from
one-shot noncoherent network coding with errors, including an overview of rank-metric
codes applied to network coding. In Section~\ref{sec:multishot} we introduce multishot
matrix codes. Section~\ref{sec:multilevel} presents a brief review of the general
multilevel theory to construct block codes over arbitrary metric spaces. The main
contribution is Section~\ref{sec:multishot-rank-metric}, where we particularize
the multilevel theory to the rank-metric space, derive an encoding procedure based
on coset partitioning, and adapt a hard-decision multistage decoding algorithm to
our problem. Finally, Section~\ref{sec:conclusion} concludes the paper.

\section{Background\label{sec:background}}

\subsection{Error Model}

This work follows the approach of Silva, Kschischang, and Koetter~\cite{silva-rank-metric},
in which the adversities of the matrix channel~\eqref{eq:matrix-channel} come in
two flavors: the rank-deficiency of the multiplicative transfer matrix and the rank
of the additive error matrix. These are constrained to\begin{equation}
\rankdef\mathbf{A}\leq\rho\qquad\mbox{and}\qquad\rank\mathbf{Z}\leq\tau,\label{eq:adversities}\end{equation}
where $\rho$ and $\tau$ are integer parameters. While $\rho$ upper bounds the
collective effect of unfortunate choices of linear combination coefficients for the
network code, wrong min-cut estimation, and packet erasures, $\tau$ measures the
maximum number of error packets injected into the network.

\subsection{Matrix Codes}

To fulfill an error-free communication over the matrix channel defined by~\eqref{eq:matrix-channel}
even under the adversities described in~\eqref{eq:adversities}, matrix codes are
employed. A (one-shot) matrix code~$\mathcal{X}$ is a non-empty subset of $\mathbb{F}_{q}^{N\times T}$.
Its \emph{code} \emph{rate} is given by\[
R(\mathcal{X})=\frac{\log_{q}\left|\mathcal{X}\right|}{NT},\]
with $0\leq R(\mathcal{X})\leq1$.

Let $\mathbf{X}$ be any codeword of a matrix code $\mathcal{X}\subseteq\mathbb{F}_{q}^{N\times T}$
and $\mathbf{Y}$ be the corresponding channel output according to~\eqref{eq:matrix-channel}.
A code~$\mathcal{X}$ is said to be \emph{$(\rho,\tau)$-correcting} if $\mathbf{X}$
can be unambiguously determined from $\mathbf{Y}$ for all choices of $\mathbf{A}$
and $\mathbf{Z}$ subject to~\eqref{eq:adversities}.

\subsection{Subspace Coding}

In \cite{silva-rank-metric}, Silva~\textsl{et~al.} obtained a sufficient condition
for the success of one-shot matrix codes under the presumed error model. Their result
asserts that a one-shot matrix code $\mathcal{X}$ is \emph{$(\rho,\tau)$-}correcting
if $d_{\mathrm{S}}(\left\langle \mathcal{X}\right\rangle )>2(2\tau+\rho).$ In this
inequality,\[
\left\langle \mathcal{X}\right\rangle =\{\left\langle \mathbf{X}\right\rangle :\mathbf{X}\in\mathcal{X}\},\]
(where $\left\langle \mathbf{X}\right\rangle $ stands for the vector subspace spanned
by the rows of matrix~$\mathbf{X}$) is the\emph{ subspace code} obtained from~$\mathcal{X}$,
while $d_{\mathrm{S}}(\left\langle \mathcal{X}\right\rangle )$ is the \emph{minimum
subspace distance} of $\left\langle \mathcal{X}\right\rangle $. The\emph{ subspace
distance} between two subspaces $U$ and $V$ is defined as \begin{equation}
d_{\mathrm{S}}(U,V)=\dim(U\dotplus V)-\dim(U\cap V)\label{eq:subspace-distance}\end{equation}
where $U\dotplus V$ is the sum subspace and $U\cap V$ is the intersection subspace.
This result reinforces the idea of \emph{transmission via subspace selection} (\textsl{i.e.},
subspace coding), as proposed in~\cite{koetter-kschischang}.

We note that in~\cite{silva-metrics}, Silva and Kschischang obtained a \emph{necessary
and sufficient} condition for a matrix code $\mathcal{X}$ to be $(\rho,\tau)$-correcting,
namely, $d_{\mathrm{I}}(\left\langle \mathcal{X}\right\rangle )>2\tau+\rho$, where
$d_{\mathrm{I}}(\cdot,\cdot)$ is called the \emph{injection distance}. Nevertheless,
we still stick to the subspace distance for mathematical simplicity. Moreover, the
proposed multilevel construction gives rise to a multishot matrix code with constant-dimension
spanned subspaces, in which case the injection distance and the subspace distance
are essentially the same.

\subsection{Rank-Metric Approach to Network Coding}

In~\cite{silva-rank-metric}, Silva~\textsl{et~al.} proposed a method to design
one-shot matrix codes based on \emph{rank-metric codes}~\cite{gabidulin,roth-crisscross}.
A rank-metric code is a block code $\mathcal{R}\subseteq\mathbb{F}_{q^{M}}^{N}$
in which the metric of concern is the \emph{rank distance} (as opposed to the Hamming
distance). The rank distance between $\mathbf{u},\mathbf{v}\in\mathbb{F}_{q^{M}}^{N}$
is defined as \begin{equation}
d_{\mathrm{R}}(\mathbf{u},\mathbf{v})=\rank(\underline{\mathbf{v}}-\underline{\mathbf{u}}),\label{eq:rank-distance}\end{equation}
where $\underline{\mathbf{u}}\in\mathbb{F}_{q}^{N\times M}$ is the matrix whose
rows are the $M$-tuples representing each of the elements of $\mathbf{u}\in\mathbb{F}_{q^{M}}^{N}$
according to some fixed basis for $\mathbb{F}_{q^{M}}$ over $\mathbb{F}_{q}$; the
rank distance is indeed a metric~\cite{gabidulin}. A matrix code $\mathcal{X}\subseteq\mathbb{F}_{q}^{N\times T}$
can be obtained from a rank-metric code $\mathcal{R}\subseteq\mathbb{F}_{q^{M}}^{N}$
by means of a simple operation called \emph{lifting}, denoted by $\mathcal{I}(\cdot)$
and defined by%
\footnote{Note that this definition differs from that of Silva \textsl{et al.}~\cite{silva-rank-metric},
where the lifting of a matrix $\mathbf{U}$ is the vector subspace spanned by $[\mathbf{I}|\mathbf{U}]$.%
}\begin{eqnarray*}
\mathcal{I}:\mathbb{F}_{q^{M}}^{N} & \longrightarrow & \mathbb{F}_{q}^{N\times T}\\
\mathbf{u} & \longmapsto & [\mathbf{I}|\underline{\mathbf{u}}],\end{eqnarray*}
where $\mathbf{I}$ is the $N\times N$ identity matrix and $T=N+M$.

The significance of the lifting operation resides in the fact that, for $\mathbf{u},\mathbf{v}\in\mathbb{F}_{q^{M}}^{N}$,
\[
d_{\mathrm{S}}(\left\langle \mathcal{I}(\mathbf{u})\right\rangle ,\left\langle \mathcal{I}(\mathbf{v})\right\rangle )=2d_{\mathrm{R}}(\mathbf{u},\mathbf{v}).\]
Thus, if $\mathcal{R}\subseteq\mathbb{F}_{q^{M}}^{N}$ is a rank-metric code then
$\mathcal{X}=\mathcal{I}(\mathcal{R})\subseteq\mathbb{F}_{q}^{N\times T}$ (obtained
by lifting each codeword of $\mathcal{R}$) gives rise to a matrix code~$\mathcal{X}$
with $\left|\mathcal{X}\right|=\left|\mathcal{R}\right|$ and $d_{\mathrm{S}}(\left\langle \mathcal{X}\right\rangle )=2d_{\mathrm{R}}(\mathcal{R})$~\cite{silva-rank-metric}.

The problem of obtaining good one-shot codes for network coding could then be reduced
to finding good rank-metric codes. But that latter task was already investigated
by Gabidulin~\cite{gabidulin}. He first proved that the Singleton bound is also
valid for rank-metric codes, that is, for every $[N,K,D]$ linear rank-metric code
over~$\mathbb{F}_{q^{M}}$ the condition $D\leq N-K+1$ must hold; he called codes
achieving this bound {}``maximum rank distance codes.'' Then he constructed a family
of maximum rank distance codes, provided $N\leq M$. Even more, Gabidulin described
encoding and decoding algorithms for his constructed codes, mainly based on the Reed-Solomon
coding theory \cite[Chapter~7]{lin-costello}.

Again in~\cite{silva-rank-metric}, Silva \textsl{et al.} adapted Gabidulin's decoding
algorithm for the matrix channel~\eqref{eq:matrix-channel}. In other words, a method
is proposed to solve the following problem: given a received matrix $\mathbf{Y}\in\mathbb{F}_{q}^{T\times M}$
and a Gabidulin code $\mathcal{R}\subseteq\mathbb{F}_{q^{M}}^{N}$, find $\hat{\mathbf{u}}\in\mathcal{R}$
such that\[
\mathbf{\hat{\mathbf{u}}}=\underset{\mathbf{u}\in\mathcal{R}}{\argmin\ }d_{\mathrm{S}}(\left\langle \mathbf{X}\right\rangle ,\left\langle \mathbf{Y}\right\rangle )\]
where $\mathbf{X}=\mathcal{I}(\mathbf{u})$. The first step of the method is to decompose
matrix $\mathbf{Y}$ into a triplet $(\mathbf{r},\mathbf{L}',\mathbf{E}')\in\mathbb{F}_{q^{M}}^{N}\times\mathbb{F}_{q}^{N\times\mu}\times\mathbb{F}_{q}^{\delta\times M}$,
operation therein called \emph{reduction} (which can be interpreted as the opposite
of lifting). Decoding then proceeds similarly to standard Gabidulin decoding of $\mathbf{r}$
for code~$\mathcal{R}$, except that \emph{side information} $\mathbf{L}'$ and
$\mathbf{E}'$ is passed to the decoder. For more details (such as the meanings of
matrices $\mathbf{L}'$ and $\mathbf{E}'$) we refer the reader to~\cite{silva-rank-metric}.

\section{Multishot Matrix Codes\label{sec:multishot}}

We finally introduce multishot matrix codes for error control in noncoherent network
coding.

\subsection{Motivation}

One of the basic problems in the realm of one-shot matrix coding is to find codes
with good rates and good error-correcting capabilities. To achieve both goals simultaneously,
it may be unavoidable to increase the field size, $q$, or the packet size, $T$.
Multishot codes allow for a third possibility: to increase the number of channel
uses, $n$. 

With that in mind, multishot codes are attractive when the system under consideration
is such that it is not possible to change the field and packet size. This is true,
for example, in fast-topology changing networks (such as wireless ones), where the
transfer matrix doesn't stay the same for much long. Under this circumstance, to
obtain codes with better error-correcting capabilities we must spread redundancy
across multiple shots. Put another way, when errors occur in a random fashion and
$q$ and $T$ are fixed, we must use the matrix channel many times in order to approach
the channel capacity.

For a simple example in which a multishot code is capable of detecting more errors
when compared with the best one could do by simply repeating one-shot codes, we point
the reader to~\cite{multishot}.

\subsection{Model and Definitions}

In this work we adopt a \emph{block-coding} approach, in which the matrix channel~\eqref{eq:matrix-channel}
is used $n$ times in a row. Our channel model then becomes\begin{equation}
\mathbf{Y}_{j}=\mathbf{A}_{j}\mathbf{X}_{j}+\mathbf{Z}_{j},\label{eq:matrix-channel-ext}\end{equation}
for $j=0,\ldots,n-1$, with matrices $\mathbf{X}_{j}$, $\mathbf{Y}_{j}$, $\mathbf{A}_{j}$,
and $\mathbf{Z}_{j}$ retaining their dimensions from the one-shot case. We allow
the adversities to be spread in any way among the $n$ time slots:

\begin{equation}
\sum_{j=0}^{n-1}\rankdef\mathbf{A}_{j}\leq\rho\qquad\mbox{and}\qquad\sum_{j=0}^{n-1}\rank\mathbf{Z}_{j}\leq\tau.\label{eq:adversities-ext}\end{equation}

An $n$-\emph{shot} (or \emph{multishot})\emph{ matrix code} $\boldsymbol{\mathcal{X}}$
is a non-empty subset of $(\mathbb{F}_{q}^{N\times T})^{n}$. Its \emph{code} \emph{rate}
is defined as the ratio between the amount of information symbols conveyed by the
transmission of a codeword and the amount of physical symbols spent by each codeword,
that is,\[
R(\boldsymbol{\mathcal{X}})=\frac{\log_{q}\left|\boldsymbol{\mathcal{X}}\right|}{nNT};\]
we have $0\leq R(\boldsymbol{\mathcal{X}})\leq1$. Akin to the one-shot case, a multishot
code $\boldsymbol{\mathcal{X}}$ is said to be \emph{$(\rho,\tau)$-correcting} if
$(\mathbf{X}_{0},\ldots,\mathbf{X}_{n-1})$ can be unambiguously determined from
$(\mathbf{Y}_{0},\ldots,\mathbf{Y}_{n-1})$ for all choices of $(\mathbf{A}_{0},\ldots,\mathbf{A}_{n-1})$
and $(\mathbf{Z}_{0},\ldots,\mathbf{Z}_{n-1})$ subject to~\eqref{eq:adversities-ext}.

\subsection{The Extended Subspace Distance}

We extend the subspace distance to\[
d_{\mathrm{S}}(\boldsymbol{U},\boldsymbol{V})=\sum_{j=0}^{n-1}d_{\mathrm{S}}(U_{j},V_{j}),\]
where $\boldsymbol{U}=(U_{0},\ldots,U_{n-1})$ and $\boldsymbol{V}=(V_{0},\ldots,V_{n-1})$
are $n$-tuples of subspaces of $\mathbb{F}_{q}^{N}$ and $d_{\mathrm{S}}(\cdot,\cdot)$
in the right-hand side is given by~\eqref{eq:subspace-distance}. We now state a
multishot counterpart for the result of Silva~\textsl{et.~al} presented in Section~\ref{sec:background}.
\begin{thm}
Let $\boldsymbol{\mathcal{X}}$ be a multishot matrix code. If $d_{\mathrm{S}}(\left\langle \boldsymbol{\mathcal{X}}\right\rangle )>2(2\rho+\tau)$
then $\boldsymbol{\mathcal{X}}$ is $(\rho,\tau)$-correcting.\end{thm}
\begin{proof}
The proof is a simple generalization of~\cite[Theorem~1]{silva-rank-metric}. Let
$\boldsymbol{X}=(\mathbf{X}_{0},\ldots,\mathbf{X}_{n-1})$ be the transmitted codeword
and $\boldsymbol{Y}=(\mathbf{Y}_{0},\ldots,\mathbf{Y}_{n-1})$ be the received sequence,
according to~\eqref{eq:matrix-channel-ext}. Let $\left\langle \boldsymbol{X}\right\rangle =(\left\langle \mathbf{X}_{0}\right\rangle ,\ldots,\left\langle \mathbf{X}_{n-1}\right\rangle )$
and $\left\langle \boldsymbol{Y}\right\rangle =(\left\langle \mathbf{Y}_{0}\right\rangle ,\ldots,\left\langle \mathbf{Y}_{n-1}\right\rangle )$
. We have\begin{gather*}
\begin{array}{l}
{\displaystyle d_{\mathrm{S}}(\left\langle \boldsymbol{X}\right\rangle ,\left\langle \boldsymbol{Y}\right\rangle )=\sum_{j=0}^{n-1}d_{\mathrm{S}}(\left\langle \mathbf{X}_{j}\right\rangle ,\left\langle \mathbf{Y}_{j}\right\rangle )}\\
{\displaystyle \quad\leq\sum_{j=0}^{n-1}d_{\mathrm{S}}(\left\langle \mathbf{X}_{j}\right\rangle ,\left\langle \mathbf{A}_{j}\mathbf{X}_{j}\right\rangle )+\sum_{j=0}^{n-1}d_{\mathrm{S}}(\left\langle \mathbf{A}_{j}\mathbf{X}_{j}\right\rangle ,\left\langle \mathbf{Y}_{j}\right\rangle )}\\
{\displaystyle \quad\leq\sum_{j=0}^{n-1}\rankdef\mathbf{A}_{j}+2\sum_{j=0}^{n-1}\rank\mathbf{Z}_{j}}\\
{\displaystyle \quad\leq\rho+2\tau.}\end{array}\end{gather*}
Since $\rho+2\tau<d_{\mathrm{S}}(\left\langle \boldsymbol{\mathcal{X}}\right\rangle )/2$,
a minimum extended subspace distance decoder is guaranteed to yield $\left\langle \boldsymbol{X}\right\rangle $
given $\left\langle \boldsymbol{Y}\right\rangle $.
\end{proof}

\subsection{Multishot Rank-Metric Codes}

Let $\boldsymbol{u}=(\mathbf{u}_{0},\ldots,\mathbf{u}_{n-1})$ and $\boldsymbol{v}=(\mathbf{v}_{0},\ldots,\mathbf{v}_{n-1})$
be two $n$-tuples of vectors in $\mathbb{F}_{q^{M}}^{N}$. The \emph{extended rank
distance} between them is defined by\[
d_{\mathrm{R}}(\boldsymbol{u},\boldsymbol{v})=\sum_{j=0}^{n-1}d_{\mathrm{R}}(\mathbf{u}_{j},\mathbf{v}_{j}),\]
where $d_{R}(\cdot,\cdot)$ in the right-hand side is the rank distance as defined
in~\eqref{eq:rank-distance}. Just like regular rank-metric codes, multishot rank-metric
codes can be applied to noncoherent network coding. To this end, we use an extended
version of the lifting operation defined as $\mathcal{I}(\boldsymbol{u})=(\mathcal{I}(\mathbf{u}_{0}),\ldots,\mathcal{I}(\mathbf{u}_{n-1})),$
where $\boldsymbol{u}=(\mathbf{u}_{0},\ldots,\mathbf{u}_{n-1})$. It is straightforward
to show that\[
d_{\mathrm{S}}(\left\langle \mathcal{I}(\boldsymbol{u})\right\rangle ,\left\langle \mathcal{I}(\boldsymbol{v})\right\rangle )=2d_{\mathrm{R}}(\boldsymbol{u},\boldsymbol{v}).\]
Thus, a multishot rank-metric code $\boldsymbol{\mathcal{R}}\subseteq(\mathbb{F}_{q^{M}}^{N})^{n}$
gives rise to a multishot matrix code $\boldsymbol{\mathcal{X}}\subseteq(\mathbb{F}_{q}^{N\times T})^{n}$
(where $T=N+M$) defined by $\boldsymbol{\mathcal{X}}=\mathcal{I}(\boldsymbol{\mathcal{R}})=\{\mathcal{I}(\boldsymbol{u}):\boldsymbol{u}\in\boldsymbol{\mathcal{R}}\}$,
with $\left|\boldsymbol{\mathcal{X}}\right|=\left|\boldsymbol{\mathcal{R}}\right|$
and $d_{\mathrm{S}}(\left\langle \boldsymbol{\mathcal{X}}\right\rangle )=2d_{\mathrm{R}}(\boldsymbol{\mathcal{R}})$.

\section{General Multilevel Code Construction\label{sec:multilevel}}

The \emph{multilevel code construction} was proposed by Imai and Hirakawa~\cite{imai-hirakawa}
in~1977 and became very popular in the 80's and 90's with more general constructions
being developed by many other researchers. Although originally targeted at codes
over a given signal set $\mathcal{S}$ of the Euclidean space, the construction can
be generalized for block codes over any finite subset~$\mathcal{S}$ of a given
metric space~$\mathcal{M}$ with associated distance~$d_{\mathrm{M}}(\cdot,\cdot)$.
(This is true as long as the component codes are Hamming-metric.) Next, we base our
description of the multilevel construction on the work of Lin and Costello~\cite[Chapter~19]{lin-costello},
wherein many references on this subject are listed.

Given a set~$\mathcal{S}$, an \emph{$m$-level partitioning} of $\mathcal{S}$
is defined by a sequence of $m+1$ partitions $\Gamma_{0},\ldots,\Gamma_{m}$ of~$\mathcal{S}$
such that $\Gamma_{0}=\{\mathcal{S}\}$ and, for $1\leq i\leq m$, partition~$\Gamma_{i}$
is a \emph{refinement} of partition~$\Gamma_{i-1}$, in the sense that the subsets
in $\Gamma_{i}$ are subsubsets of the subsets in $\Gamma_{i-1}$. It is possible
to represent an $m$-level partitioning by a \emph{rooted tree} with $m+1$ levels,
labeled from $0$ to $m$. The nodes at level~$i$ are the subsets in the partition
$\Gamma_{i}$. The unique node at level $0$ is called the the \emph{root node} (which
is the set~$\mathcal{S}$) while the nodes at level $m$ are called the \emph{leaf
nodes}. A node $\mathcal{Y}\in\Gamma_{i}$ is a \emph{child} of the only element
$\mathcal{X}\in\Gamma_{i-1}$ such that $\mathcal{Y}\subseteq\mathcal{X}$. Equivalently,
a node $\mathcal{Y}\in\Gamma_{i}$ is the \emph{parent} of every node $\mathcal{Z}\in\Gamma_{i+1}$
such that $\mathcal{Z}\subseteq\mathcal{Y}$.

A level $i$ is said to be \emph{nested} if every node in this level has the same
number $p_{i}$ of children, although we do allow the partitions at level $i+1$
to have different cardinalities. (Note that, by this definition, level $0$ is always
nested in any partitioning.) In our construction of multishot codes we require level
$i$ to be nested for $0\leq i<m$. The edges joining a subset at level $i$ to subsets
at level $i+1$ in the tree can then be labeled with the numbers $0,\ldots,p_{i}-1$.
We denote by $Q(c^{(0)},\ldots,c^{(m-1)})$ the subset of nodes in $\Gamma_{m}$
reached by following the path $(c^{(0)},\ldots,c^{(m-1)})$ in the tree, where $0\leq c^{(i)}<p_{i}$
for $0\leq i<m$.

Consider now a metric space $\mathcal{M}$ with distance $d_{\mathrm{M}}(\cdot,\cdot)$
and let $\mathcal{S}\subseteq\mathcal{M}$ be a finite subset of $\mathcal{M}$.
We now describe the procedure to construct a block code~$\boldsymbol{\mathcal{C}}$
over $\mathcal{S}$ of length~$n$. Let $\Gamma_{0},\ldots,\Gamma_{m}$ be an $m$-level
partitioning of~$\mathcal{S}$, with level~$i$ nested for $0\leq i<m$. We define
the \emph{intrasubset distance} of level $i$ as\[
d_{\mathrm{M}}^{(i)}=\min\{d_{\mathrm{M}}(\mathcal{Y}):\mathcal{Y}\in\Gamma_{i}\},\]
for $0\leq i<m$. All levels $0\leq i<m$ must be {}``protected'' by classical
codes of length~$n$ over $\mathbb{F}_{p_{i}}$, called \emph{component codes} and
denoted by~$\mathcal{H}_{i}$, with minimum Hamming distances\[
d_{\mathrm{H}}^{(i)}=d_{\mathrm{H}}(\mathcal{H}_{i}).\]

The codewords of $\boldsymbol{\mathcal{C}}\subseteq\mathcal{M}^{n}$ are obtained
as follows.
\begin{enumerate}
\item Form all possible arrays of $m$~rows and $n$~columns\[
\boldsymbol{\Lambda}=\left[\begin{array}{cccc}
c_{0}^{(0)} & c_{1}^{(0)} & \cdots & c_{n-1}^{(0)}\\
c_{0}^{(1)} & c_{1}^{(1)} & \cdots & c_{n-1}^{(1)}\\
\vdots & \vdots & \ddots & \vdots\\
c_{0}^{(m-1)} & c_{1}^{(m-1)} & \cdots & c_{n-1}^{(m-1)}\end{array}\right],\]
where the $i$-th row of $\boldsymbol{\Lambda}$ is a codeword $\mathbf{c}^{(i)}=(c_{0}^{(i)},\ldots,c_{n-1}^{(i)})$
of code~$\mathcal{H}_{i}$, for $0\leq i<m$. The set of all such arrays is denoted
by~$\mathcal{A}$ and has cardinality $\left|\mathcal{A}\right|=\prod_{i=0}^{m-1}|\mathcal{H}_{i}|$.
\item The $j$-th column $\mathbf{c}_{j}=(c_{j}^{(0)},\ldots,c_{j}^{(m-1)})$ of a given
array $\mathbf{\boldsymbol{\Lambda}}\in\mathcal{A}$ specifies a path in the rooted
tree, starting at $\mathcal{S}\in\Gamma_{0}$ and ending at $Q(\mathbf{c}_{j})\in\Gamma_{m}$,
for $0\leq j\leq n$.
\item Each array $\mathbf{\boldsymbol{\Lambda}}\in\mathcal{A}$ gives rise to a set of
codewords:\[
\boldsymbol{\mathcal{C}}_{\boldsymbol{\Lambda}}=Q(\mathbf{c}_{0})\times Q(\mathbf{c}_{1})\times\cdots\times Q(\mathbf{c}_{n-1}).\]

\item Finally, the constructed code $\boldsymbol{\mathcal{C}}$ is the union of all such
(disjoint) sets:\[
\boldsymbol{\mathcal{C}}=\bigcup_{\boldsymbol{\Lambda}\in\mathcal{A}}\boldsymbol{\mathcal{C}}_{\boldsymbol{\Lambda}}.\]

\end{enumerate}
Accordingly, the total number of codewords in the constructed code $\boldsymbol{\mathcal{C}}$
is\[
\left|\boldsymbol{\mathcal{C}}\right|=\sum_{\boldsymbol{\Lambda}\in\mathcal{A}}\ \prod_{0\leq j<n}\left|Q(\mathbf{c}_{j})\right|.\]
For the case when $\left|\mathcal{Y}\right|=1$ for all $\mathcal{Y}\in\Gamma_{m}$,
each array $\boldsymbol{\Lambda}\in\mathcal{A}$ gives rise to exactly one codeword
and the encoding procedure can be represented by Figure~\ref{fig:Diagram1}. In
this case,\begin{equation}
\left|\boldsymbol{\mathcal{C}}\right|=\left|\mathcal{A}\right|=\prod_{i=0}^{m}\left|\mathcal{H}_{i}\right|.\label{eq:multilevel-cardinality-1}\end{equation}
Also, from multilevel theory~\cite{lin-costello}, the minimum (extended) distance
of the constructed code~$\boldsymbol{\mathcal{C}}$ is lower-bounded by\begin{equation}
d_{\mathrm{M}}(\boldsymbol{\mathcal{C}})\geq\min\{d_{\mathrm{M}}^{(i)}d_{\mathrm{H}}^{(i)}:0\leq i<m\}.\label{eq:multilevel-distance}\end{equation}
\begin{figure}
\begin{centering}
\psfrag{H0}[c][B][0.9][0]{$\mathcal{H}_0$}
\psfrag{H1}[c][B][0.9][0]{$\mathcal{H}_1$}
\psfrag{Hm}[c][B][0.9][0]{$\mathcal{H}_{m-1}$}
\psfrag{vdots}[][][0.9][0]{$\vdots$}
\psfrag{c0}[c][][0.9][0]{$\mathbf{c}^{(0)}$}
\psfrag{c1}[c][][0.9][0]{$\mathbf{c}^{(1)}$}
\psfrag{cm}[c][][0.9][0]{$\mathbf{c}^{(m-1)}$}
\psfrag{buff}[c][][0.9][90]{Buffer}
\psfrag{cj}[c][][0.9][0]{$(\mathbf{c}_j)_{0 \leq j < n}$}
\psfrag{map}[c][][0.9][0]{$Q(\cdot)$}\includegraphics[scale=0.3]{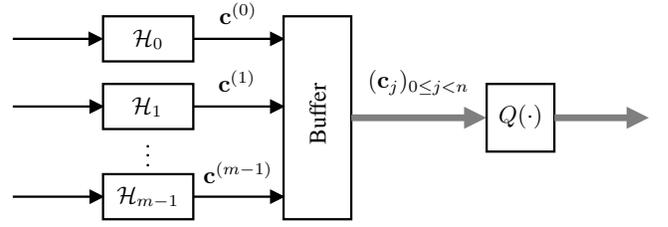}
\par\end{centering}

\centering{}\caption{Block diagram for a general multilevel encoder.}
\label{fig:Diagram1}
\end{figure}

\section{Multilevel Construction using\protect \\
Rank-Metric Codes\label{sec:multishot-rank-metric}}

In this section, we aim at constructing a multishot rank-metric code $\boldsymbol{\mathcal{R}}\subseteq(\mathbb{F}_{q^{M}}^{N})^{n}$
which, by lifting of each component of its codewords, can be converted to a multishot
block matrix code $\boldsymbol{\mathcal{X}}\subseteq(\mathbb{F}_{q}^{N\times T})^{n}$
with $T=N+M$ (\textsl{cf.} Section~\ref{sec:multishot}). To this end, we particularize
the multilevel construction described earlier to the case where the metric space~$\mathcal{M}$
with $d_{\mathrm{M}}(\cdot,\cdot)$ is, in fact, the space $\mathbb{F}_{q^{M}}^{N}$
with the rank distance $d_{\mathrm{R}}(\cdot,\cdot)$. The finite subset $\mathcal{S}\subseteq\mathcal{M}$
will then be a $q^{M}$-ary $[N,K,D]$ linear rank-metric code $\mathcal{R}\subseteq\mathbb{F}_{q^{M}}^{N}$.
Furthermore, the multilevel partitioning $\Gamma_{0},\ldots,\Gamma_{m}$ will be
obtained by a technique called \emph{coset partitioning} \cite[Section~4.5]{lin-costello},
described next.

\begin{figure*}
\begin{centering}
\psfrag{H0}[cc][B][1.0][0]{Enc $\mathcal{H}_0$}
\psfrag{H1}[cc][B][1.0][0]{Enc $\mathcal{H}_1$}
\psfrag{Hm}[cc][B][1.0][0]{Enc $\mathcal{H}_{m-1}$}
\psfrag{vdots}[][][1.0][0]{$\vdots$}
\psfrag{c0}[cc][][1.0][0]{$\mathbf{c}^{(0)} \equiv (\mathbf{m}^{(0)}_j)_{0 \leq j < n}$}
\psfrag{c1}[cc][][1.0][0]{$\mathbf{c}^{(1)} \equiv (\mathbf{m}^{(1)}_j)_{0 \leq j < n}$}
\psfrag{cm}[cc][][1.0][0]{$\mathbf{c}^{(m-1)} \equiv (\mathbf{m}^{(m-1)}_j)_{0 \leq j < n}$}
\psfrag{coset0}[cc][B][1.0][0]{Enc $[\mathcal{R}_0 / \mathcal{R}_1]$}
\psfrag{coset1}[cc][B][1.0][0]{Enc $[\mathcal{R}_1 / \mathcal{R}_2]$}
\psfrag{cosetm}[cc][B][1.0][0]{Enc $[\mathcal{R}_{m-1} / \mathcal{R}_m]$}
\psfrag{v0}[cc][][1.0][0]{$(\mathbf{v}^{(0)}_j)_{0 \leq j < n}$}
\psfrag{v1}[cc][][1.0][0]{$(\mathbf{v}^{(1)}_j)_{0 \leq j < n}$}
\psfrag{vm}[cc][][1.0][0]{$(\mathbf{v}^{(m-1)}_j)_{0 \leq j < n}$}
\psfrag{u}[cc][][1.0][0]{$(\mathbf{u}_j)_{0 \leq j < n}$}\includegraphics[scale=0.4]{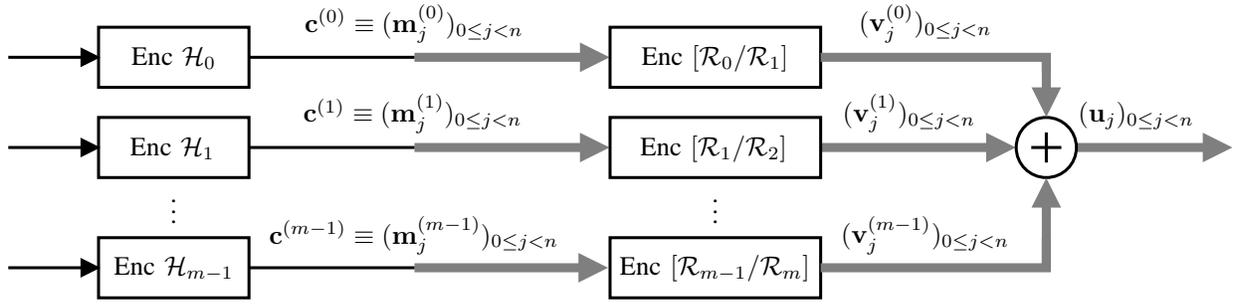}
\par\end{centering}

\caption{Block diagram for a generalized concatenated coding system.}
\label{fig:Diagram2}
\end{figure*}

\subsection{Coset Partitioning}

Let\[
\mathbf{G}=\left[\begin{array}{cccc}
g_{0,0} & g_{0,1} & \cdots & g_{0,K-1}\\
g_{1,0} & g_{1,1} & \cdots & g_{1,K-1}\\
\vdots & \vdots & \ddots & \vdots\\
g_{N-1,0} & g_{N-1,1} & \cdots & g_{N-1,K-1}\end{array}\right]\]
be a generating matrix%
\footnote{We adopt the convention that the code is the \emph{column} space of~$\mathbf{G}$.%
} for $\mathcal{R}$ and let $K=K_{0}>\cdots>K_{m}=0$ be a strictly decreasing sequence
of integers. For $0\leq i\leq m$, define $\mathcal{R}_{i}$ to be the linear code
generated by $\mathbf{G}[0:K_{i}-1]$ (\textsl{i.e.}, the first $K_{i}$ columns
of~$\mathbf{G}$) and $\bar{\mathcal{R}}_{i}$ to be the linear code generated by
$\mathbf{G}[K_{i}:K-1]$ (\textsl{i.e.}, the last $K-K_{i}$ columns of~$\mathbf{G}$).
Then\[
\Gamma_{i}=\{\mathcal{R}_{i}+\mathbf{v}:\mathbf{v}\in\bar{\mathcal{R}}_{i}\}\]
defines an $m$-level partitioning $\Gamma_{0},\ldots,\Gamma_{m}$ of $\mathcal{R}$
with minimum intrasubset distances $d_{\mathrm{R}}^{(i)}=d_{\mathrm{R}}(\mathcal{R}_{i})$,
for $0\leq i<m$. In other words, $\Gamma_{i}$ consists of cosets of $\mathcal{R}_{i}$
having as coset leaders the elements of $\bar{\mathcal{R}}_{i}$. Note that we have
$\Gamma_{0}=\{\mathcal{R}\}$ and $\Gamma_{m}=\{\{\mathbf{u}\}:\mathbf{u}\in\mathcal{R}\}$.
Additionally, each level $i$, for $0\leq i<m$, is nested, with every node at level~$i$
having \[
p_{i}=\frac{\left|\Gamma_{i+1}\right|}{\left|\Gamma_{i}\right|}=\frac{(q^{M})^{K-K_{i+1}}}{(q^{M})^{K-K_{i}}}=q^{M(K_{i}-K_{i+1})}\]
children.

\subsection{Generalized Concatenation and Encoding Procedure}

The multilevel partitioning $\Gamma_{0},\ldots,\Gamma_{m}$ just constructed, along
with suitable $p_{i}$-ary Hamming-metric component codes $\mathcal{H}_{0},\ldots,\mathcal{H}_{m-1}$
gives rise to a multishot rank-metric code $\boldsymbol{\mathcal{R}}\subseteq(\mathbb{F}_{q^{M}}^{N})^{n}$
with cardinality given by \eqref{eq:multilevel-cardinality-1} and minimum distance
satisfying \eqref{eq:multilevel-distance}. Nevertheless, the construction---as presented
in Section~\ref{sec:multilevel}---does not specify any efficient encoding or decoding
procedure. To this end, we will make use of the connection between the multilevel
coding theory and \emph{generalized concatenated codes} (also known as \emph{multilevel
concatenated codes})~\cite[Chapter~15]{lin-costello}.

Let $\mathbf{c}_{0},\ldots,\mathbf{c}_{n-1}$ be the output of the buffer in Figure~\ref{fig:Diagram1}.
As said before, each $\mathbf{c}_{j}$ represents a path in the rooted tree starting
at $\mathcal{R}\in\Gamma_{0}$ and ending at the leaf $\{\mathbf{u}_{j}\}=Q(\mathbf{c}_{j})\in\Gamma_{m}$.
Recall that $p_{i}=q^{M(K_{i}-K_{i+1})}$. In view of that, each component $c_{j}^{(i)}\in\mathbb{F}_{p_{i}}$
of $\mathbf{c}_{j}$ can also be viewed as a $(K_{i}-K_{i+1})$-tuple with elements
in~$\mathbb{F}_{q^{M}}$. Denote this tuple by $\mathbf{m}_{j}^{(i)}\in\mathbb{F}_{q^{M}}^{K_{i}-K_{i+1}}$.
This suggests us to define\[
\mathbf{u}_{j}\triangleq\sum_{i=0}^{m-1}\mathbf{G}[K_{i+1}:K_{i}-1]\cdot\mathbf{m}_{j}^{(i)}=\sum_{i=0}^{m-1}\mathbf{v}_{j}^{(i)},\]
where each $\mathbf{v}_{j}^{(i)}=\mathbf{G}[K_{i+1}:K_{i}-1]\cdot\mathbf{m}_{j}^{(i)}$
can be viewed as the codeword associated with message $\mathbf{m}_{j}^{(i)}$ of
the linear code generated by matrix $\mathbf{G}[K_{i+1}:K_{i}-1]$ (\textsl{i.e.},
column $K_{i+1}$ up to, and including, column $K_{i}-1$ of $\mathbf{G}$). This
code is denoted by $[\mathcal{R}_{i}/\mathcal{R}_{i+1}]$ and, since it contains
coset leaders for partition $\mathcal{R}_{i}/\mathcal{R}_{i+1}$ (\textsl{i.e.},
$\mathcal{R}_{i+1}\subseteq\mathcal{R}_{i}$ and its cosets), it is called a \emph{coset
code} \cite[Section~15.2]{lin-costello}.

Thus, encoding can be summarized in the following steps, illustrated in Figure~\ref{fig:Diagram2}.
\begin{enumerate}
\item Let $\mathbf{c}^{(i)}\in\mathcal{H}_{i}$ be a codeword of~$\mathcal{H}_{i}$, for
$0\leq i<m$.
\item Translate each codeword $\mathbf{c}^{(i)}=(c_{0}^{(i)},\ldots,c_{n-1}^{(i)})$ into
$(\mathbf{m}_{0}^{(i)},\ldots,\mathbf{m}_{n-1}^{(i)})$.
\item Encode each $\mathbf{m}_{j}^{(i)}\in\mathbb{F}_{q^{M}}^{K_{i}-K_{i+1}}$, $0\leq j<n$,
using $\mathbf{G}[K_{i+1}:K_{i}-1]$ to form $\mathbf{v}_{j}^{(i)}\in\mathbb{F}_{q^{M}}^{N}$.
\item Finally, the $j$-th coordinate of the codeword is calculated according to $\mathbf{u}_{j}=\sum_{i=0}^{m-1}\mathbf{v}_{j}^{(i)}$.
\end{enumerate}

In the terminology of generalized concatenated codes, the component codes $\mathcal{H}_{i}$
are called \emph{outer codes}, while the coset codes $[\mathcal{R}_{i}/\mathcal{R}_{i+1}]$
are the \emph{inner codes}.

\subsection{A Special Situation}

Consider the special situation in which $m=K$ and $K_{i}=K-i$ in a way that $p_{i}=q^{M}$
for $0\leq i<m$. If, in addition, (\textsl{i})~every rank-metric code $\mathcal{R}_{i}$
is maximum rank distance, and (\textsl{ii})~every ($q^{M}$-ary) component code
$\mathcal{H}_{i}$ is maximum Hamming distance separable with distance $d_{\mathrm{H}}^{(i)}=\left\lceil d/d_{\mathrm{R}}^{(i)}\right\rceil $,
then we have that \[
d_{\mathrm{R}}^{(i)}=N-K_{i}+1=N-K+i+1\]
and\[
\log_{q^{M}}\left|\mathcal{H}_{i}\right|=n-\left\lceil \frac{d}{d_{\mathrm{R}}^{(i)}}\right\rceil +1.\]

The first condition is always achievable with Gabidulin codes whenever $N\leq M$
(this becomes clear from the structure of a generating matrix for Gabidulin codes~\cite{gabidulin}).
The second condition is also achievable if $n<q^{M}$ (\textsl{e.g.}, with Reed-Solomon
codes), which is typically true. Hence, in view of~\eqref{eq:multilevel-distance}
and~\eqref{eq:multilevel-cardinality-1}, we get a multishot rank-metric code~$\boldsymbol{\mathcal{R}}$
with minimum distance $d_{\mathrm{R}}(\boldsymbol{\mathcal{R}})=d$ and cardinality\begin{eqnarray}
\log_{q}\left|\boldsymbol{\mathcal{R}}\right| & = & \sum_{i=0}^{m-1}\log_{q}\left|\mathcal{H}_{i}\right|\label{eq:bound}\\
 & = & \sum_{i=0}^{K-1}M\left(n+1-\left\lceil \frac{d}{N-K+i+1}\right\rceil \right)\nonumber \\
 & = & MK(n+1)-M\sum_{i=0}^{K-1}\left\lceil \frac{d}{N-K+i+1}\right\rceil ,\nonumber \end{eqnarray}
this value---after maximized over all $K\in\{0,\ldots,N\}$---being an upper bound
on the size of any multishot rank-metric code constructed using the proposed method.

\begin{figure*}
\begin{centering}
\psfrag{r}[cc][][1.0][0]{$(\mathbf{r}^{(i)}_j)_{0 \leq j < n}$}
\psfrag{LE}[cc][][1.0][0]{$(\mathbf{L}'_j, \mathbf{E}'_j)_{0 \leq j < n}$}
\psfrag{u_tilde}[cc][][1.0][0]{$(\tilde{\mathbf{u}}^{(i)}_j)_{0 \leq j < n}$}
\psfrag{v_tilde}[cc][][1.0][0]{$(\tilde{\mathbf{v}}^{(i)}_j)_{0 \leq j < n}$}
\psfrag{m_tilde}[cc][][1.0][0]{$(\tilde{\mathbf{m}}^{(i)}_j)_{0 \leq j < n} \equiv \tilde{\mathbf{c}}^{(i)}$}
\psfrag{c_hat}[cc][][1.0][0]{$\hat{\mathbf{c}}^{(i)}$}
\psfrag{decR}[cc][B][1.0][0]{Dec $\mathcal{R}_i$}
\psfrag{decH}[cc][B][1.0][0]{Dec $\mathcal{H}_i$}
\psfrag{coset_id}[cc][B][1.0][0]{Coset Leader}
\psfrag{inversemap}[cc][B][1.0][0]{$\mathrm{Enc}^{-1}$ $[\mathcal{R}_i / \mathcal{R}_{i+1}]$}\includegraphics[scale=0.4]{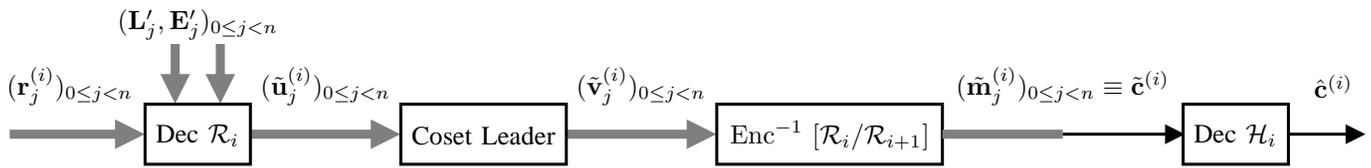}
\par\end{centering}

\caption{Block diagram for concatenated decoder at $i$-th stage.}
\label{fig:Diagram3}
\end{figure*}

\subsection{Hard-Decision Multistage Decoding}

We now suggest a sub-optimal hard-decision multistage decoding algorithm \cite[Section~15.2]{lin-costello}
for the case when all $\mathcal{R}_{i}$ are Gabidulin codes. Let $\boldsymbol{\mathcal{X}}\subseteq(\mathbb{F}_{q}^{N\times T})^{n}$
be a multishot matrix code obtained by the multilevel construction just described.
Following the model in Section~\ref{sec:multishot}, let $(\mathbf{X}_{0},\ldots,\mathbf{X}_{n-1})\in\boldsymbol{\mathcal{X}}$
be the transmitted codeword and $(\mathbf{Y}_{0},\ldots,\mathbf{Y}_{n-1})\in(\mathbb{F}_{q}^{N\times T})^{n}$
be the received sequence.

The multistage decoding occurs in $m$ stages; we start by finding the reduction
$(\mathbf{r}_{j},\mathbf{L}_{j}',\mathbf{E}_{j}')\in\mathbb{F}_{q^{M}}^{N}\times\mathbb{F}_{q}^{N\times\mu}\times\mathbb{F}_{q}^{\delta\times M}$
of each $\mathbf{Y}_{j}\in\mathbb{F}_{q}^{T\times M}$ and setting $\mathbf{r}_{j}^{(0)}=\mathbf{r}_{j}$.
The decoding then proceeds in an iterative fashion. For $0\leq i<m$, at the $i$-th
stage, the following steps are executed (Figure~\ref{fig:Diagram3}).
\begin{enumerate}
\item \emph{Inner decoding}. Using the generalized decoding method of Silva\textsl{ et
al.}~\cite{silva-rank-metric}, decode $(\mathbf{r}_{j}^{(i)},\mathbf{L}_{j}',\mathbf{E}_{j}')$
into a codeword $\tilde{\mathbf{u}}_{j}^{(i)}\in\mathcal{R}_{i}$. Of course $\tilde{\mathbf{u}}_{j}^{(i)}$
belongs to some coset in $\mathcal{R}_{i}/\mathcal{R}_{i+1}$. Identify the leader~$\tilde{\mathbf{v}}_{j}^{(i)}\in[\mathcal{R}_{i}/\mathcal{R}_{i+1}]$
of this coset and find---by inverse mapping---the message $\tilde{\mathbf{m}}_{j}^{(i)}\in\mathbb{F}_{q^{M}}^{K_{i}-K_{i+1}}$
that generated it.
\item After $n$ inner decodings we obtain $(\tilde{\mathbf{m}}_{0}^{(i)},\ldots,\tilde{\mathbf{m}}_{n-1}^{(i)})$.
Similarly to the encoding procedure, define $\tilde{\mathbf{c}}^{(i)}$ as the vector
in $\mathbb{F}_{p_{i}}^{n}$ corresponding to $(\tilde{\mathbf{m}}_{0}^{(i)},\ldots,\tilde{\mathbf{m}}_{n-1}^{(i)})$.
\item \emph{Outer decoding}. Note that, due to errors, $\tilde{\mathbf{c}}^{(i)}$ may
not be a codeword of $\mathcal{H}_{i}$. Therefore, decode $\tilde{\mathbf{c}}^{(i)}$
into the closest (in Hamming sense) codeword $\hat{\mathbf{c}}^{(i)}\in\mathcal{H}_{i}$.
\end{enumerate}

At the end of the iteration, each $\mathbf{r}_{j}^{(i)}$ is updated to\[
\mathbf{r}_{j}^{(i+1)}=\mathbf{r}_{j}^{(i)}-\hat{\mathbf{v}}_{j}^{(i)},\]
where each $\hat{\mathbf{v}}_{j}^{(i)}$ is obtained from $\hat{\mathbf{c}}^{(i)}$
according to $[\mathcal{R}_{i}/\mathcal{R}_{i+1}]$ (just like in Steps~2 and~3
of the encoding procedure) and decoding proceeds to the next stage. After $m$ steps,
we have $\hat{\mathbf{c}}^{(0)},\ldots,\hat{\mathbf{c}}^{(m-1)}$. Figure~\ref{fig:Diagram4}
illustrates the whole process.

\begin{figure}
\begin{centering}
\psfrag{r}[cc][][1.0][0]{$(\mathbf{r}_j)_{0 \leq j < n}$}
\psfrag{r0}[cc][][1.0][0]{$(\mathbf{r}^{(0)}_j)_{0 \leq j < n}$}
\psfrag{r1}[cc][][1.0][0]{$(\mathbf{r}^{(1)}_j)_{0 \leq j < n}$}
\psfrag{r2}[cc][][1.0][0]{$(\mathbf{r}^{(2)}_j)_{0 \leq j < n}$}
\psfrag{rm}[cc][][1.0][0]{$(\mathbf{r}^{(m-1)}_j)_{0 \leq j < n}$}
\psfrag{vdots}[][][1.0][0]{$\vdots$}
\psfrag{-}[][][1.0][0]{$-$}
\psfrag{c0}[cc][][1.0][0]{$\hat{\mathbf{c}}^{(0)}$}
\psfrag{c1}[cc][][1.0][0]{$\hat{\mathbf{c}}^{(1)}$}
\psfrag{cm}[cc][][1.0][0]{$\hat{\mathbf{c}}^{(m-1)}$}
\psfrag{concat0}[cc][B][1.0][0]{ConcatDec $0$}
\psfrag{concat1}[cc][B][1.0][0]{ConcatDec $1$}
\psfrag{concatm}[cc][B][1.0][0]{ConcatDec $m-1$}
\psfrag{coset0}[cc][B][1.0][0]{Enc $[\mathcal{R}_0 / \mathcal{R}_1]$}
\psfrag{coset1}[cc][B][1.0][0]{Enc $[\mathcal{R}_1 / \mathcal{R}_2]$}
\psfrag{v0}[cc][][1.0][0]{$(\hat{\mathbf{v}}^{(0)}_j)_{0 \leq j < n}$}
\psfrag{v1}[cc][][1.0][0]{$(\hat{\mathbf{v}}^{(1)}_j)_{0 \leq j < n}$}\includegraphics[scale=0.4]{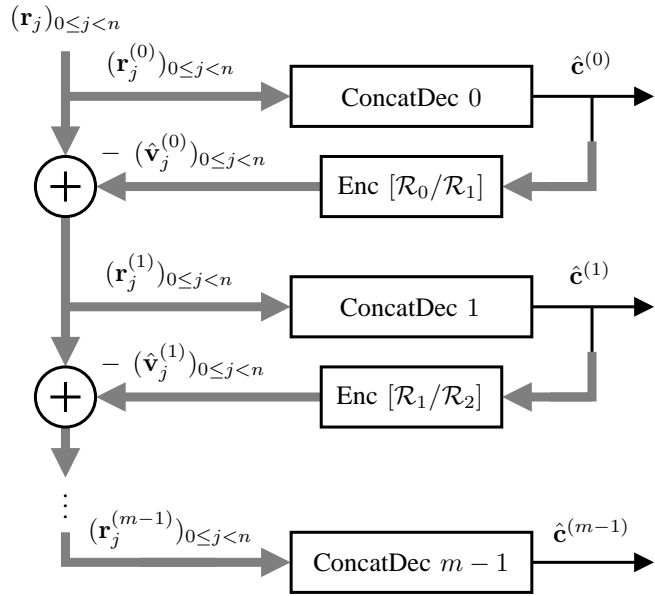}
\par\end{centering}

\caption{Block diagram for hard-decision multistage decoder. Each one of the $\mathrm{ConcatDec}$
blocks is detailed in Figure~\ref{fig:Diagram3}.}
\label{fig:Diagram4}
\end{figure}

\section{Conclusion\label{sec:conclusion}}

This work presented a bit more explicit multilevel construction of multishot codes
for network coding than that introduced in~\cite{multishot}. A natural question
arises: how good are the proposed codes? This demands a comparison of~\eqref{eq:bound}
with known bounds or previous one-shot constructions.

Another open problem is to adapt soft-decision multistage decoding algorithms~\cite[Section~15.3]{lin-costello}
to the current scenario. In particular, the metric is now the rank distance (as opposed
to the Euclidean distance in the case of codes for the AWGN channel). This can possibly
take advantage of list decoding of rank-metric codes.

\section*{Acknowledgment}

The authors would like to thank Danilo Silva for helpful discussions, and CAPES
(Brazil) and CNPq (Brazil) for bibliographical research and financial support.

\bibliographystyle{ieeetr}
\bibliography{../biblio}

\end{document}